\documentclass{llncs}
\usepackage{graphicx} 
\usepackage{comment}
\usepackage{complexity}
\usepackage{xspace}
\usepackage[most]{tcolorbox}
\usepackage{bold-extra}
\usepackage{hyperref}
\usepackage{todonotes}
\usepackage{adjustbox}
\newtheorem{observation}[theorem]{Observation}

\usepackage[english]{babel}
\usepackage[pagewise]{lineno}
\usepackage[utf8]{inputenc}
\usepackage[autostyle, english = american]{csquotes}
\MakeOuterQuote{"}
\usepackage{amssymb}
\usepackage{amsmath}
\usepackage{algorithm}
\usepackage{algpseudocode}
\usepackage{lmodern}
\newcommand{\problemdef}[3]{
\vspace{0.3cm}
\begin{tcolorbox}[enhanced,attach boxed title to top left={yshift=-3.5mm,yshifttext=0mm,xshift=3mm}, colback=gray!3,colframe=black,colbacktitle=white,boxrule=0.5pt,
  title=\textsc{#1},coltitle=black,fonttitle=\bfseries,
  boxed title style={size=small,colframe=black, boxrule=0.5pt,colback=white} ]
  \vspace{0.3cm}
  \textbf{Given:} \hspace{1.5em} #2 \\[4pt]
  \textbf{Question:} \hspace{0em} #3
\end{tcolorbox}
\vspace{0.3cm}
}

\usetikzlibrary{decorations.pathreplacing}
\usetikzlibrary{positioning, shapes}

\begin{document}

\title{On the Complexity of Claw-Free Vertex Splitting}

\titlerunning{On the Complexity of Claw-Free Vertex Splitting}

\author{Faisal N. Abu-Khzam \and
Sergio Thoumi}
\authorrunning{F. N. Abu-Khzam and S. Thoumi}

\institute{Department of Computer Science and Mathematics\\ 
Lebanese American University\\
Beirut, Lebanon\\
\email{\{faisal.abukhzam,sergio.thoumi\}@lau.edu.lb}}
\maketitle

\begin{abstract}
Vertex splitting consists of taking a vertex $v$ in a graph and replacing it with two non-adjacent vertices whose combined neighborhoods is the neighborhood of $v$. The split is said to be exclusive when these neighborhoods are disjoint.
In the \textsc{Claw-Free (Exclusive) Vertex Splitting} problem, we are given a graph $G$ and an integer $k$, and we are asked if we can perform at most $k$ (exclusive) vertex splits to obtain a claw-free graph. 
We consider the complexity of \textsc{Claw-Free Exclusive Vertex Splitting} and prove it to be $\NP$-complete in general, while admitting a polynomial-time algorithm when the input graph has maximum degree 4. This result settles an open problem posed in [Firbas \& Sorge, ISAAC 2024]. We also show that our results can be generalized to \textsc{$K_{1,c}$-Free Vertex Splitting} for all $c \geq 3$. 
\end{abstract}

\section{Introduction}

A graph property $\Pi$ is \emph{hereditary} if, whenever a graph $G$
satisfies $\Pi$, every induced subgraph of $G$ also satisfies it.
A hereditary property is non-trivial if infinitely many graphs satisfy it and infinitely many graphs do not. Applying a sequence of modifications to obtain a graph with a given (non-trivial) hereditary property is a well-studied type of problem in graph algorithms. 
Numerous properties and modification operations have been considered. For the \emph{vertex deletion} operation, the problem of finding a minimum set of vertices whose deletion results in a graph that satisfies $\Pi$ is $\NP$-complete for \emph{any} non-trivial hereditary property $\Pi$ \cite{lewis1980node}. A classical example is the \textsc{Feedback Vertex Set} problem, which corresponds to vertex deletion into a forest. Moreover, \emph{vertex deletion} into hereditary properties is fixed-parameter tractable ($\FPT$) if the hereditary property can be characterized by a finite set of forbidden induced subgraphs \cite{cai1996fixed}. In contrast, when the allowed operation is \emph{edge deletion}, the computational complexity varies depending on the target property. For instance, finding a minimum set of edges to delete to obtain a forest is polynomial-time solvable \cite{koch2024edge}. However, edge deletion into a cactus graph is $\NP$-complete \cite{koch2024edge}. 

Recently, vertex splitting into hereditary properties ($\Pi$-VS) started to garner more interest \cite{nollenburg2025planarizing,firbas2023establishing,firbas2024complexity,baumann2024parameterized,splittingbipartite,planar2layer}. Vertex splitting takes a vertex $v$ and replaces it by two new non-adjacent vertices $v_1,v_2$ such that $N(v)=N(v_1)\cup N(v_2)$, where $N(v)$ is the set of vertices adjacent to $v$. If the new neighborhoods are disjoint (i.e. $N(v_1) \cap N(v_2) = \emptyset$), the operation is known as an \emph{exclusive} vertex splitting. Otherwise, it is said to be 
\emph{inclusive}. 
Exclusive vertex splits are particularly of interest due to their usage in application domains such as correlation clustering \cite{AbuIsTh2025,amin}. As with edge deletion, the complexity of the problem was found to depend on the specific property under consideration \cite{firbas2023establishing,firbas2024complexity}. 
While the complexity of $\Pi$-VS has been established for most hereditary properties, the problem remains unclassified for some important properties. One of the remaining open problems is splitting into a claw-free graph (\textsc{Claw-Free Vertex Splitting}) \cite{firbas2024complexity}. A graph is said to be \emph{claw-free} if it does not contain an \emph{induced} subgraph that has a vertex with three pairwise non-adjacent neighbors. 

In this paper, we investigate the complexity of splitting into $K_{1,c}$-free graphs. This is a generalization of \textsc{Claw-Free Vertex Splitting} which corresponds to the case where $c=3$. Our first main result establishes that \emph{exclusive vertex splitting} into a $K_{1,c}$-free graph is $\NP$-complete for all $c \geq 3$. We also present a polynomial-time algorithm for \textsc{$K_{1,3}$-Free Vertex Splitting} on bounded degree 4 graphs. We believe the presented results take us a step closer towards defining the jagged line between easy and hard problems in $\Pi$-VS.

\section{Preliminaries}

We work with simple, undirected, unweighted finite graphs. We use standard terminology as presented in \cite{diestel}.

A $K_{1,3}$ is a graph formed of a vertex with three pairwise non-adjacent neighbors. It is also known as a \emph{claw}. In this case, the degree 3 vertex is known as a \emph{claw center} and each of its 3 neighbors is said to be \emph{leaf}. A graph is said to be \emph{claw-free} if it does not have an induced subgraph that is a claw. In this context, a claw is said to be an induced \emph{forbidden subgraph}.

As mentioned earlier, \emph{vertex splitting} is an operation that takes a vertex $v$ and replaces it with two new non-adjacent vertices $v_1$ and $v_2$ such that $N(v) = N(v_1) \cup N(v_2)$. When the two neighborhoods are disjoint, i.e., $N(v_1) \cap N(v_2) = \emptyset$, the operation is referred to as an \emph{exclusive vertex split}. Otherwise, it is known as an \emph{inclusive vertex split}. 

While this paper is about the complexity of splitting into a claw-free graph, we note that vertex splitting could sometimes introduce new claws. This can happen only when a vertex $v$ has three neighbors $x,y$ and $z$ such that 
$y$ and $z$ are non-adjacent and $N(x) \cap \{y,z\} \neq \emptyset$.
In this case, splitting $x$ so that one copy, say $x_1$, is adjacent to $v$ but $N(x_1)\cap \{y,z\} = \emptyset$ results in a claw centered at $v$.
Obviously, this is the only way a new claw can be introduced since splitting a vertex does not make it the center of a new claw.
%
%
%


We now formally define the \textsc{$K_{1,c}$-Free Vertex Splitting} problem parameterized by $k$, where $c$ is treated as a fixed constant.

\problemdef{\textsc{$K_{1,c}$-Free Vertex Splitting}}{A graph $G$ and a positive integer $k$.}{Can we perform at most $k$ vertex splitting operations such that we obtain a $K_{1,c}$-free graph?}

When $c=3$, the problem is known as \textsc{Claw-Free Vertex Splitting} (CVFS). We will generally refer to the problem as \textsc{$K_{1,c}$-Free Vertex Splitting} when referring to arbitrary fixed values of $c \geq 3$. Moreover, we refer to the variant of this problem where we consider only exclusive (resp., inclusive) vertex splits as \textsc{Claw-Free Exclusive Vertex Splitting} (resp., \textsc{Claw-Free Inclusive Vertex Splitting}). The complexity of each of these variants remains an open problem \cite{firbas2023establishing,firbas2024complexity}. We now begin by presenting the below simple observations that will be used throughout the paper:

\begin{observation}\label{obs:disconnectclaw}
A vertex split can remove a claw only if it is applied to its center vertex.
\end{observation}

\begin{observation} \label{obs:cliquecontribution}
A claw can have at most one leaf vertex from a clique.
\end{observation}

The {\sc Claw-Free Vertex Splitting} problem is obviously solvable in polynomial time on graphs of maximum degree 3 or less. In fact, when a claw center is of degree exactly 3, then its 3 neighbors must be pairwise non-adjacent. In this case any splitting that results in a vertex of degree 2 and a pendant vertex can be applied, and it is obviously optimal (cannot create extra splits).

\section{Bounded Degree Graphs} \label{alg:deg4}

 \begin{theorem}
\textsc{Claw-Free Vertex Splitting} is solvable in polynomial-time on graphs of maximum degree 4.
\end{theorem}

\begin{proof}

The key step in the presented approach is to split a claw center $v$ in such a way that no extra splits are created. If a single split simultaneously eliminates all existing claws at $v$ and does not create extra splits, 
%
then the split is obviously optimal. Although multiple distinct splits may satisfy these conditions, we present one such split for each case below.

\subsubsection*{Case 1: $G[N(v)]$ has 1 connected component (see Figure \ref{fig:polyonecc}).}

In this case, without loss of generality, we split $v$ such that $N(v_1)=\{a\}$ and $N(v_2)=\{b,c,d\}$. Note that here the vertex $a$ is already a claw center with leaves $\{b,d,c\}$. Although the split will create an additional claw centered at $a$ since $G[N(a)]$ will now have 4 (pendant) connected components, the graph will still only need one extra split to become claw-free (see Case \ref{fig:polyfourcc}).

\begin{figure}[h]
    \centering
\begin{tikzpicture}[every node/.style={circle, draw, minimum size=8mm}, node distance=2cm]
  \node (v) at (0,0) {$v$};
  \node (a) [below=of v] {$a$};
  \node (b) [below left=1.2cm and 1.5cm of v] {$b$};
  \node (c) [below right=1.2cm and 1.5cm of v] {$c$};
  \node (d) [left=2cm of v] {$d$};

  \draw (v) -- (a);
  \draw (v) -- (b);
  \draw (v) -- (c);
  \draw (v) -- (d);

  \draw (a) -- (b);
  \draw (a) -- (c);
\draw[bend left=65] (a) to (d);
\end{tikzpicture}
    \caption{Case 1.}
    \label{fig:polyonecc}
\end{figure}

\vspace{5pt}

For the remaining cases, whenever have a claw center \(v\), we can always find a pendant vertex in \(G[N(v)]\). This enables us to split \(v\) in a way that avoids removing any edge that is part of a triangle. As a result, the performed split does not create extra splits.

\subsubsection*{Case 2: $G[N(v)]$ has 2 connected components (see Figure \ref{fig:polytwocc}).}

\noindent

Here, exactly 2 of the 3 dashed edges should be present in the graph to ensure $G[N(v)]$ has 2 connected components while $v$ is not claw-free. If we have less than 2, then we will have 3 connected components. If we have all 3 edges, we will not have a claw since $\{a,b,c\}$ will induce a clique (see Observation \ref{obs:cliquecontribution}). Regardless of which 2 edges are present, we can split $v$ such that $N(v_1)=\{p\}$ and $N(v_2)=\{a,b,c\}$.

\begin{figure}[h]
    \centering
    \begin{tikzpicture}[every node/.style={circle, draw, minimum size=8mm}, node distance=2cm]
  \node (v) at (0,0) {$v$};
  \node (p) [above=of v] {$p$};
  \node (a) [below left=1.2cm and 1.5cm of v] {$a$};
  \node (b) [below=1.8cm of v] {$b$};
  \node (c) [below right=1.2cm and 1.5cm of v] {$c$};

  \draw (v) -- (p);
  \draw (v) -- (a);
  \draw (v) -- (b);
  \draw (v) -- (c);

  \draw[dashed] (a) -- (b);
  \draw[dashed] (b) -- (c);
  \draw[dashed] (a) -- (c);
\end{tikzpicture}

    \caption{Case 2.}
    \label{fig:polytwocc}
\end{figure}

\subsubsection*{Case 3: $G[N(v)]$ has 3 connected components (see Figure \ref{fig:polythreecc}).}

\noindent
We can simply split $v$ such that $N(v_1)=\{p_1\}$ and $N(v_2)=\{a,b,p_2\}$.

\begin{figure}[h]
    \centering
\begin{tikzpicture}[every node/.style={circle, draw, minimum size=8mm}, node distance=2cm]
  \node (v) at (0,0) {$v$};
  \node (p1) [above left=of v] {$p_1$};
  \node (p2) [above right=of v] {$p_2$};
  \node (a) [below left=1.2cm and 1.5cm of v] {$a$};
  \node (b) [below right=1.2cm and 1.5cm of v] {$b$};

  \draw (v) -- (p1);
  \draw (v) -- (p2);
  \draw (v) -- (a);
  \draw (v) -- (b);

  \draw (a) -- (b);
\end{tikzpicture}

    \caption{Case 3.}
    \label{fig:polythreecc}
\end{figure}

\subsubsection*{Case 4: $G[N(v)]$ has 4 connected components (see Figure \ref{fig:polyfourcc}).}

\noindent
In this case, all connected components must be pendant vertices. We split $v$ such that $N(v_1)=\{p_1,p_2\}$ and $N(v_2)=\{p_3,p_4\}$.

\begin{figure}[h]
    \centering
\begin{tikzpicture}[every node/.style={circle, draw, minimum size=8mm}, node distance=2cm]
  \node (v) at (0,0) {$v$};
  \node (p1) [above left=of v] {$p_1$};
  \node (p2) [above right=of v] {$p_2$};
  \node (a) [below left=1.2cm and 1.5cm of v] {$p_3$};
  \node (b) [below right=1.2cm and 1.5cm of v] {$p_4$};

  \draw (v) -- (p1);
  \draw (v) -- (p2);
  \draw (v) -- (a);
  \draw (v) -- (b);

\end{tikzpicture}

    \caption{Case 4.}
    \label{fig:polyfourcc}
\end{figure}

\vspace{5pt}

\end{proof}
 
A successive application of the above steps is guaranteed to result in a claw-free graph. Moreover, since we use the minimum number of splits (one per claw-center), our algorithm solves the CFVS problem (not only the exclusive vertex-splitting variant).

\section{Complexity of {\sc Claw-Free Exclusive Vertex Splitting}}

The approach described above for obtaining a polynomial-time algorithm on graphs with maximum degree 4 does not necessarily extend to graphs with maximum degree more than 4. In such graphs, we may have a vertex $v$ of degree at least 5 for which $G[N(v)]$ does not contain a pendant vertex. Furthermore, in some instances, (locally) eliminating a claw may require multiple vertex splits, thereby increasing the difficulty of designing a polynomial-time algorithm.

We prove that {\sc Claw-Free Exclusive Vertex Splitting} (CFEVS) is $\NP$-complete. Our proof is based on a reduction from the {\sc Cubic Vertex Cover} problem, which is formally defined as follows:.


\problemdef{\textsc{Cubic Vertex Cover (CVC)}}
{A positive integer $t$ and a 
3-regular graph $G=(V,E)$;}{Does $G$ have a vertex cover of size at most $t$?}

The above problem is $\NP$-complete (See Fact 7.2 in \cite{greenlaw1995cubic}). We further assume that $t>2$.

\begin{theorem}
   \textsc{Claw-Free Exclusive Vertex Splitting} is 
   $\NP$-hard.
\end{theorem}

\begin{proof}

Let $(G,t)$ be an arbitrary instance of CVC, we construct an equivalent instance $(G',k)$ of CFEVS as follows:

\begin{itemize}
    \item[-] For each vertex $v_i \in V(G)$, we create a corresponding \emph{link vertex} $l_i$. Let $L = \{l_i \mid v_i \in V(G)\}$ denote the set of all link vertices.

    \item[-] Each link vertex $l_i \in L$ is adjacent to a pendant vertex $p_i$.
    
    \item[-] For each edge $e_i \in E(G)$, we create a corresponding \emph{main vertex} $m_i$. Let $M = \{m_i \mid e_i \in E(G)\}$ be the set of all main vertices. We add an edge between every pair of vertices in $M$, so that the subgraph $G'[M]$ is a clique.

    \item[-] We add an edge between a link vertex $l_i$ and a main vertex $m_j$ if and only if the vertex $v_i$ (corresponding to $l_i$) is an endpoint of the edge $e_j$ (corresponding to $m_j$) in $G$.

    \item[-] For each $i  \in [1, |E(G)| + t + 1]$, we create a \emph{help vertex} $h_i$. Let $H = \{h_i \mid 1 \le i \le |E(G)| + t + 1\}$. We add an edge between every pair of vertices in $H$, so that $G'[H]$ is a clique.

    \item[-] Each help vertex $h_i \in H$ is adjacent to all main vertices in $M$. Additionally, it is adjacent to a pendant vertex $z_i$.
    
    \item[-] We let $k=|E(G)|+t$.
\end{itemize}

An example of the construction is illustrated in Figure \ref{fig:constructionk13proof}.

\begin{figure}[htb!]
\centering
\begin{tikzpicture}[
    scale = 1,
    vertex/.style={circle, draw, minimum size=7mm},
    node distance=2cm
]

\node[vertex] (1) at (0,2) {$l_1$};
\node[vertex] (2) at (2,2) {$l_2$};
\node[vertex] (3) at (4,2) {$l_3$};
\node[vertex] (4) at (6,2) {$l_4$};

\node[vertex] (p1) at (0,3) {$p_1$};
\node[vertex] (p2) at (2,3) {$p_2$};
\node[vertex] (p3) at (4,3) {$p_3$};
\node[vertex] (p4) at (6,3) {$p_4$};

\draw[thick] (1) -- (p1);
\draw[thick] (2) -- (p2);
\draw[thick] (3) -- (p3);
\draw[thick] (4) -- (p4);

\node[] (E) at (-3,2) {$L$};
\node[] (A) at (-3,0) {$M$};
\node[] (B) at (-3,-2) {$H$};
\node[vertex] (s1) at (0,0) {$m_1$};
\node[vertex] (s2) at (1.2,0) {$m_2$};
\node[vertex] (s3) at (2.4,0) {$m_3$};
\node[vertex] (s4) at (3.6,0) {$m_4$};
\node[vertex] (s5) at (4.8,0) {$m_5$};
\node[vertex] (s6) at (6,0) {$m_6$};

\draw[thick] (s1) -- (s2);
\draw[thick] (s2) -- (s3);
\draw[thick] (s3) -- (s4);
\draw[thick] (s4) -- (s5);
\draw[thick] (s5) -- (s6);
\draw[thick] (s1) to[bend right=30] (s3);
\draw[thick] (s1) to[bend right=40] (s4);
\draw[thick] (s1) to[bend right=50] (s5);
\draw[thick] (s1) to[bend right=55] (s6);
\draw[thick] (s2) to[bend right=30] (s4);
\draw[thick] (s2) to[bend right=40] (s5);
\draw[thick] (s2) to[bend right=50] (s6);
\draw[thick] (s3) to[bend right=30] (s5);
\draw[thick] (s3) to[bend right=40] (s6);
\draw[thick] (s4) to[bend right=30] (s6);

\draw[red, thick] (s1) -- (1);
\draw[red, thick] (s1) -- (2);
\draw[blue, thick] (s2) -- (1);
\draw[blue, thick] (s2) -- (3);
\draw[green, thick] (s3) -- (1);
\draw[green, thick] (s3) -- (4);
\draw[orange, thick] (s4) -- (2);
\draw[orange, thick] (s4) -- (3);
\draw[purple, thick] (s5) -- (2);
\draw[purple, thick] (s5) -- (4);
\draw[brown, thick] (s6) -- (3);
\draw[brown, thick] (s6) -- (4);

\node[vertex] (t0) at (-1.2,-2) {$h_1$};
\node[vertex] (t1) at (0,-2) {$h_2$};
\node[vertex] (t2) at (1.2,-2) {$h_3$};
\node[vertex] (t3) at (2.4,-2) {$h_4$};
\node[vertex] (t4) at (3.6,-2) {$h_5$};
\node[vertex] (t5) at (4.8,-2) {$h_6$};
\node[vertex] (t6) at (6,-2) {$h_7$};
\node[vertex] (t7) at (7.2,-2) {};

\node[vertex] (pt0) at (-1.2,-4.1) {$z_1$};
\node[vertex] (pt1) at (0,-4.1) {$z_2$};
\node[vertex] (pt2) at (1.2,-4.1) {$z_3$};
\node[vertex] (pt3) at (2.4,-4.1) {$z_4$};
\node[vertex] (pt4) at (3.6,-4.1) {$z_5$};
\node[vertex] (pt5) at (4.8,-4.1) {$z_6$};
\node[vertex] (pt6) at (6,-4.1) {$z_6$};
\node[vertex] (pt7) at (7.2,-4.1) {};

\foreach \i in {0,...,7} {
  \draw[thick] (t\i) -- (pt\i);
}

\foreach \i in {0,...,7} {
  \foreach \j in {1,...,6} {
    \draw[gray, thin] (t\i) -- (s\j);
  }
}

\foreach \i/\j in {0/1,1/2,2/3,3/4,4/5,5/6} {
  \draw[blue, thick] (t\i) -- (t\j);
}
\draw[blue, thick, dashed] (t6) -- (t7);

\draw[blue, thick] (t0) to[bend right=20] (t2);
\draw[blue, thick] (t0) to[bend right=25] (t3);
\draw[blue, thick] (t0) to[bend right=30] (t4);
\draw[blue, thick] (t0) to[bend right=35] (t5);
\draw[blue, thick] (t0) to[bend right=40] (t6);
\draw[blue, thick] (t0) to[bend right=41] (t7);

\draw[blue, thick] (t1) to[bend right=20] (t3);
\draw[blue, thick] (t1) to[bend right=25] (t4);
\draw[blue, thick] (t1) to[bend right=30] (t5);
\draw[blue, thick] (t1) to[bend right=35] (t6);
\draw[blue, thick] (t1) to[bend right=40] (t7);

\draw[blue, thick] (t2) to[bend right=20] (t4);
\draw[blue, thick] (t2) to[bend right=25] (t5);
\draw[blue, thick] (t2) to[bend right=30] (t6);
\draw[blue, thick] (t2) to[bend right=35] (t7);

\draw[blue, thick] (t3) to[bend right=20] (t5);
\draw[blue, thick] (t3) to[bend right=25] (t6);
\draw[blue, thick] (t3) to[bend right=30] (t7);

\draw[blue, thick] (t4) to[bend right=20] (t6);
\draw[blue, thick] (t4) to[bend right=25] (t7);

\draw[blue, thick] (t5) to[bend right=20] (t7);

\coordinate (brace_start) at ([xshift=-5pt]t0.south west);
\coordinate (brace_end) at ([xshift=5pt]t7.south east);
\draw [decorate, decoration={brace,amplitude=10pt,mirror,raise=4pt}]
    (brace_start |- 0,-4.25) -- (brace_end |- 0,-4.25)
    node [black,midway,yshift=-0.75cm] {$|E(G)|+t+1$};

\end{tikzpicture}
\caption{Construction for the instance where the given graph is a $K_4$.}
\label{fig:constructionk13proof}
\end{figure}

\vspace{5pt}
\noindent
{\bf Claim.}
A given instance of \textsc{CVC} is a YES-instance if and only if its corresponding constructed instance of \textsc{CFEVS} is also a YES-instance.

\vspace{5pt}
\noindent
($\Rightarrow$)
Observe that each main vertex $m_i$ needs to split since we have claws centered at it (see Observation \ref{obs:disconnectclaw}). 
Initially, these are the only claw centers in the graph. Let $VC$ be a minimal vertex cover of size at most $t$. We split each main vertex $m_i$ into two vertices $m_{i1}$ (copy one) and $m_{i2}$ (copy two) so that $m_{i1}$ is adjacent to only one vertex from $L$ representing an element of $VC$, while 
$m_{i2}$ will be adjacent to all the remaining vertices in $N(m_i)$. We repeat this process for every main vertex while making sure that a vertex of $L$ is adjacent to the same copy (one or two) whenever we split a main vertex. For example, if $l_1\in VC$, then each time we split a main vertex $m_i$ that is adjacent to $l_1$, we should have $N(m_{i1})=l_1$ and $N(m_{i2})=N(m_{i})\setminus \{l_1\}$. When the edge represented by a main vertex has both of its endpoints in $VC$, the neighbor of its copy-one vertex can be any of its two elements. 

In total, we have used $|M|=|E(G)|$ splits until now, resulting in $|E(G)|$ additional vertices that form an independent set of pendant vertices, namely $A'=\{m_{i1}: i\in M\}$. The unique neighbor of each element of $A'$ is a vertex that corresponds to an element of $VC$, as per the above splitting scenario. Observe also that the (at most) $t$ vertices that correspond to the elements of $VC$ become claw centers after the splitting operations performed so far, and these are the only claw centers in the resulting graph. Since each of these centers is of degree exactly 4, each claw can be disconnected by performing exactly one split (for example, by creating a copy that is adjacent to the pendant vertex $p_i$ only, as in the algorithm described in Section \ref{alg:deg4}). 

\vspace{5pt}
\noindent
($\Leftarrow$) First, let us prove that we can have a solution to the CFEVS instance with at most $|E(G)|+t$ splits. Initially, we need at least $|E(G)|$ splits since there are claws centered at each main vertex (see Observation \ref{obs:disconnectclaw}).
If a main vertex $m_i$ is split into $m_{i1}$ and $m_{i2}$, then let $M_j = N_M(m_{ij})$ and $H_j = N_H(m_{ij})$, $j \in \{1,2\}$. Of course, $M_1\cap M_2 = H_1\cap H_2 = \emptyset$
(simply because we are dealing with exclusive splits). We note the following:

\noindent
(i) If none of the sets $H_1$ and $H_2$ is empty, then each vertex $h_i\in H_j$ is a claw center with leaves $m_{ij}$ along with ``any'' element of $H_{3-j}$ and the pendant neighbor of $h_i$, which yields a NO-instance. Therefore, without loss of generality, we may assume that $H_1 = \emptyset$ and $H_2 = H$.

\noindent
(ii) If none of the sets $M_1$ and $M_2$ is empty then each vertex $h_i\in H$ is a claw center with leaves $m_{i2}$ along with ``any'' element of $M_1$ and the pendant neighbor of $H$, which again yields a NO-instance. Therefore, without loss of generality, we may assume that $M_1 = \emptyset$. In other words, one copy of $m_i$, which we chose to be $m_{i1}$, must be totally disconnected from $M\cup H$. This means that none of the (remaining) elements of $M\cup H$ will become a claw center after the (described) split of any of the main vertices. From this point on, only the elements of $L$ can become claw centers. 

\noindent
(iii) Assume the $i^{th}$ edge of $G$, represented by vertex $m_i \in G'$ has $\{e_{i1}, e_{i2}\}$ at its endpoints. Then either $m_{i1}$ is adjacent to a single element of this subset or to both of them.
In the latter case, 2 claws are obtained from the split of $m_i$, with centers $e_{i1}$ and $e_{i2}$. Since the total number of remaining splits cannot exceed $t$, we assume that we have at most $t$ elements of $L$ that are claw centers. 
Without loss of optimality, and to actually minimize the number of remaining claw centers, we can dictate that each vertex $m_{i1}$ is adjacent to exactly one element from $L$. Therefore, for the CFEVS instance to be a YES-instance, there must be at most $t$ elements of $L$ that can be split and each vertex $m_{i1}$ must be adjacent to one of them. It follows that the (at most) $t$ claw centers of $L$ form a vertex cover of $G$.  
\end{proof}

Since the reduction is done in polynomial time and the problem is obviously in $\NP$, then:

\begin{corollary}
\textsc{$K_{1,3}$-Free Exclusive Vertex Splitting} is $\NP$-complete.
\end{corollary}

We also state the following corollary:

\begin{corollary}
 $K_{1,c}$-Free Exclusive Vertex Splitting is $\NP$-hard for all $c \geq 3$.
\end{corollary}

To prove the above corollary, simply apply the same construction from \textsc{Vertex Cover} on $c$-regular hypergraphs where each edge joins $c-1$ vertices\footnote{This problem is $\NP$-hard being a generalization of \textsc{Cubic Vertex Cover $(c=3)$}.}, except that each element of $L\cup H$ will have $c-2$ pendant neighbors.

\section{Summary}

We proved that \textsc{$K_{1,c}$-Free Exclusive Vertex Splitting} is $\NP$-complete for all $c \geq 3$. We also proved that \textsc{$K_{1,3}$-Vertex Splitting} is solvable in polynomial time on graphs of bounded degree 4. 
We leave the following questions as open problems: Is \textsc{Claw-Free Vertex Splitting} $\FPT$ with respect to the number of splits? And, if so, does it have a polynomial-size kernel?

\section*{Funding}

This research project was partially supported by the Lebanese American University under the President’s Intramural Research Fund PIRF0056.

\end{document}